\documentclass[12pt,a4paper]{article}

\usepackage{amsmath}
\usepackage{amsfonts}
\usepackage{dsfont}
\usepackage{amssymb}
\usepackage{a4}
\usepackage{graphicx}
\usepackage{amsthm}
\usepackage{multirow}
\usepackage[round]{natbib}
\usepackage{lscape}
\usepackage{rotating}
\usepackage{setspace}
\usepackage{geometry}
\usepackage[hidelinks]{hyperref}

\doublespace

\theoremstyle{plain}
\newtheorem{theorem}{Theorem}%[section]
\newtheorem{lemma}{Lemma}%[theorem]
\newcommand{\blind}{0}

\begin{document}
\date{} 
\if0\blind{

\makeatletter
\def\blfootnote{\gdef\@thefnmark{}\@footnotetext}
\makeatother

\title{Direct semi-parametric estimation of the state price density implied in option prices}

\author{Gianluca Frasso$\dag$ $^{\ast}$ and Paul H.C. Eilers ${\ddag}$}
%\author{Gianluca Frasso and Paul H.C. Eilers}
\blfootnote{{$^\ast$ Corresponding author. Email: gianluca.frasso@features-analytics.com} \\
$\dag$ Facult\'{e} des Sciences Humaines et Sociales, M\'{e}thodes Quantitatives en Sciences Sociales, Universit\'{e} de Li\`{e}ge, Belgium \&
Features-Analytics S.A., Nivelles, Belgium.
Email: gianluca.frasso@features-analytics.com \\
$\ddag$Erasmus University Medical Center, Rotterdam, The Netherlands. Email: p.eilers@erasmusmc.nl}
\maketitle
} \fi

\if1\blind
{
  \bigskip
  \bigskip
  \bigskip
  \begin{center}
    {\LARGE\bf Direct semi-parametric estimation of the state price density implied in option prices}
\end{center}
  \medskip
} \fi

\begin{abstract}
We present a model for direct semi-parametric estimation of the state price density (SPD) implied by quoted option prices. We treat the observed prices as expected values of possible pay-offs at maturity, weighted by the unknown probability density function. We model the logarithm of the latter as a smooth function, using P-splines, while matching the expected values of the potential pay-offs with the observed prices. This leads to a special case of the penalized composite link model. Our estimates do not rely on any parametric assumption on the underlying asset price dynamics and are consistent with no-arbitrage conditions. The model shows excellent performance in simulations and in applications to real data.
\end{abstract}
\section{Introduction}
\label{introduction}
Under equilibrium conditions, the value of an option contract is equal to the discounted expected value of its future net returns. The expectation is taken with respect to the state price density (SPD), also referred to as the risk-neutral density \citep{CoxandRoss1976} or equivalent martingale measure \citep{HarrisonandKreps1979}. This probability measure cannot be observed directly. However, under restrictive assumptions on the underlying asset price dynamics, its functional form is known.  This is the case within the log-normal framework proposed by \citet[][]{Black1973}. Unfortunately, log-normality is rarely appropriate \citep[see e.g.][]{Bates2000}, and more reliable pricing approaches are preferable in many applications. 

The risk-neutral density is proportional to the second derivative of the option price with respect to the strike price \citep{Breeden1978}. Option contracts are not marketed for a continuous set of strikes and observed prices can be noisy and biased. This makes a simple numerical differentiation of the observed quotes useless for inferring the underlying distribution. Parametric, semi-parametric and non-parametric models have been proposed to solve this estimation task \citep[see for example][for an extensive presentation]{Jondeau2007}.

Parametric schemes assume an analytic form of the SPD. \citet{Ritchey1990} models the risk-neutral distribution as a mixture of normal densities while \cite*{Bahra1996} and \citet{Melick1997} adopt a mixture of log-normal densities. 

Semi-parametric methods aim to approximate in a flexible way departures from a parametric density. \citet*{Jarrow1989} adopt an approach based on the Edgeworth expansion of a log-normal risk-neutral density. Within a similar framework, \citet{Abken1996} and \citet*{Madan1994} present a (Hermite) polynomial approximation, while \citet{Jondeau2001} suggest constrained Gram-Charlier expansions.

Non-parametric strategies do not formulate any hypothesis about the underlying asset dynamics. Several non-parametric proposals focus on indirect estimation of the latent distribution function by smoothing the observed prices and successively approximating its second derivative. \citet{Sahalia1998, Sahalia2000} and \citet{Huynh2002} model the observed option quotes by kernel smoothing. Within the same framework, \citet{Sahalia2003} recommend a two-stage procedure: in a first step the data are pre-processed using isotonic regression, and the fitted option values are smoothed via kernel techniques. \citet{Song2016} 
estimate the SPDs implied by S\&P500 and VIX options with a two stage kernel smoothing method 
and use them to infer the joint parametric pricing kernel. \citet{Campa1998}, \citet{Bliss2002} and \citet{Jackwerth2000} use methods that are based on the regularization of the implied volatility curves. \citet{Jackwerth1996} model the squared differences of the estimated pricing function with the observed one and maximize the smoothness of the underlying density. 

Other non-parametric proposals involve polynomial basis expansions. \citet{Shimko1993} approximates the implied volatility surface using quadratic polynomials and derives continuous option pay-offs from the smoothed volatility smiles, while \cite{Fengler2009} defines a class of constrained natural spline smoothers for the implied volatility curves. \citet{Rosenberg1998} uses sigma-shaped polynomials. \citet[][]{Yatchew2006} propose smoothing splines regularized by a penalty term forcing a no-arbitrage fit. \citet*{Hardle2009} obtain similar results with (suitably re-parameterized) constrained non-linear regression of the option prices, an approach that does not ensure smooth approximations. \citet{Bondarenko2003} proposes a positive density convolution (PCA) estimator. 

We introduce a direct and flexible framework enabling smooth estimation of the state price density implied in option contracts (DESPD: direct estimation of the state price density). Our strategy leads to a penalized composite link model \citep{Eilers2007} and complies with arbitrage-free requirements. We treat the observed prices as expected values of possible pay-offs at maturity, weighted by the latent (unknown) density. We model the logarithm of the latter as a smooth function, using P-splines, while matching the expected value of the possible contract pay-offs with the observed ones. The option prices and the SPD are estimated as functions of penalized regression coefficients. Smoothness is induced by a discrete roughness penalty on the spline coefficients, \citep[in analogy with][]{Eilers1996} allowing for efficient interpolation and extrapolation at  unobserved support points \citep[see e.g.][]{Eilers2010}.

This paper is organized as follows. In Section~\ref{sec_settings_estimation} we introduce the general setting of our approach and describe the estimation procedure. In Section~\ref{sec_noarbitrage} we prove the no-arbitrage properties of the DESPD estimator. A large simulation study is presented in Section~\ref{sec_simulations}. In section~\ref{sec_application} we illustrate the results obtained analyzing over one year of weekly option contracts written on the S\&P500 index. Concluding remarks and future research paths are discussed in Section~\ref{sec_discussion}. 

\section{DESPD model specification and estimation procedure}
\label{sec_settings_estimation}
Quoting \citet*{Cox1979}, an option is a contract giving the right (not the obligation) to buy (call type) or sell (put type) a risky asset with price $s$ at a predetermined (fixed) strike price $k$ before (American style) or at (European style) a given date (maturity of the contract, $T$). Call and put options written on a given underlying are usually quoted together for a set of strike prices over different maturities (so called \textit{option chains}). Here we consider only European-style (exercise only possible at maturity) contracts. 

The current price (at time $t$) of an option (say a call) should take into account the uncertainty about its pay-off $(s_{T} - k)$ at the expiration date (apart from the cost of money and transaction costs). Its fair value is then equal to the discounted expected returns at maturity ($T$):
\begin{equation}\label{eq:integral_price}
c_{t} = \exp(-r_{t, \tau} \tau)\int_{0}^{\infty}(s_{T} - k)^{+} f_{t}(s_{T}) 
\mbox{d}s_{T},
\end{equation}
where $r_{t, \tau}$ is the risk-free interest rate (e.g. the Libor rate), $\tau = T - t$ is the time to maturity, $s_{T}$ (price of the underlying at maturity) is the state variable and $f_{t}(s_{T})$ is the state price density. The SPD is unknown and must be inferred from the observed prices. 

Suppose we observe a set of (call) option prices $c_{1}, \dots{}, c_{n}$ with ordered strikes   $k_{1}< k_{2}< \dots < k_{n}$ (what follows can be generalized to put options, as discussed in Section~\ref{sec_parity}). Define a set of possible values for the underlying asset at maturity $\boldsymbol{u} = \{u_{1},\dots,u_{m}\}$ with $u_{1} < u_{2} < \dots < u_{m}$ and $m$ fixed a priori. We set the elements of $\boldsymbol{u}$ as equally spaced (in Section~\ref{subsec:gridvalue} we discuss this choice and present an alternative definition) and require that $k_{1} > u_{1} $ and $k_{n} < u_{m}$. In our experience $u_{1} = \max(0, k_{1} - 0.1 \times k_{1}) $ and $u_{m} = k_{n} + 0.1 \times k_{n}$ are appropriate choices in most applications. Let also $\varphi_{j}$ be the $j$th element of the $m-$dimensional vector $\boldsymbol{\varphi}$ such that $\varphi_{j}= f_{t}(u_{j}) = \displaystyle \frac{\exp({\eta}_{j})}{1 + \sum_{k = 2}^{m} \exp({\eta}_{k})}$, where $\eta_{1} = 0$ to ensure identifiability \citep[see e.g.][]{Lipovetsky2009}. To keep our notation simple, we also assume that $\exp(-r_{t, \tau} \tau) = 1$ (or that the observations have been scaled by the known discount factor). We can then model the observed prices as
\begin{equation}\label{eq:nonlinear_regression}
\displaystyle c_{i} = \displaystyle  \mu_{i} + \epsilon_{i} = \displaystyle  \sum_{j = 1}^{m} g_{i j} \varphi_{j} + \epsilon_{i},
\end{equation}
where $i = 1,...,n$, $\epsilon_{i}$ are i.i.d. random variables with zero mean and constant variance $\sigma^2$, $g_{ij}=(u_{j} - k_{i})^{+}$ (for $i = 1,\dots, n$, and $j = 1,\dots,m$) are the entries of a $(n \times m)$ matrix $\boldsymbol{G}$:
\begin{equation*}
	\boldsymbol{G} =
	\left[
	\begin{array}{c c c c c c}
	(u_{1} - k_{1})^{+} & (u_{2} - k_{1})^{+} & \cdots & (u_{j} - k_{1})^{+} & \cdots  & (u_{m} - k_{1})^{+}\\
	%(u_{1} - k_{2})^{+} & (u_{2} - k_{2})^{+} & \cdots & (u_{j} - k_{2})^{+} & \cdots  & (u_{m} - k_{2})^{+}\\
	\vdots              & \vdots              &\vdots  &   \vdots                  & \vdots  & \vdots             \\
	(u_{1} - k_{i})^{+} & (u_{2} - k_{i})^{+} & \cdots & (u_{j} - k_{i})^{+} & \cdots  & (u_{m} - k_{i})^{+}\\
	\vdots              & \vdots              &\vdots  &    \vdots                 & \vdots  & \vdots             \\
	(u_{1} - k_{n})^{+} & (u_{2} - k_{n})^{+} & \cdots & (u_{j} - k_{n})^{+} & \cdots  & (u_{m} - k_{n})^{+}
	\end{array}
	\right],
\end{equation*}
Here $(u_j - k_i)^+ = (u_j-k_i) \mathtt{I}(u_j > k_i)$. Eq.~\ref{eq:nonlinear_regression} leads to the following optimization problem
\begin{equation}\label{eq:ls_nopen}
	\min_{\boldsymbol{\eta}} {S}(\boldsymbol{\eta}) = \|\boldsymbol{c} - \boldsymbol
	{G} \boldsymbol{\varphi} \|^{2}.
\end{equation}
This is a severely ill-conditioned non-linear problem. We regularize it by assuming smoothness of the unknown probability density function and penalizing the differences of the adjacent $\boldsymbol{\eta}$ coefficients. We can then write the penalized non-linear problem as follows \citep[see e.g.][]{Eilers2007} 
\begin{equation}\label{eq:pen_LS}
\min_{\boldsymbol{\eta}} S(\boldsymbol{\eta}) = \|\boldsymbol{c} - \boldsymbol{G} \boldsymbol{\varphi} \|^{2} + \lambda \|\boldsymbol{D} \boldsymbol{\eta}\|^{2},
\end{equation}
where $\boldsymbol{D}$ is a difference operator. Second or third order differences are common choices for the penalty matrix. In our analyses we will use a third order $\boldsymbol{D}$-matrix, promoting unimodal density estimates (see e.g. \citealp*{Eilers2005}). 

The estimator in Eq.~\ref{eq:pen_LS} differs from the one proposed by \citet{Jackwerth1996} since we model the logarithm of the latent distribution as a smooth function. Our parameterization ensures proper density estimates with no need for an explicit constraint and guarantees estimates consistent with no-arbitrage requirements (see Section~\ref{sec_noarbitrage}). It is also different from the one of \citet{Hardle2009}, because smooth density estimates are naturally enforced by the roughness penalty term. The smoothness of $\boldsymbol{\varphi}$ is regulated by $\lambda$ which can be selected by means of an optimality criterion (see Section~\ref{subsec_penalty_sel}) or specified by the user.

\begin{lemma}
For a fixed $\lambda > 0$, the functions 
$L(\boldsymbol{\eta}) = \|\boldsymbol{c} - \boldsymbol{G}\boldsymbol{\varphi} \|^{2}$ and 
$Q(\boldsymbol{\eta}) = \lambda \|\boldsymbol{D} \boldsymbol{\eta}\|^{2}$ are differentiable at every point ${\eta} \in \mathds{R}^{m}$. They are twice differentiable except when $\boldsymbol{D} \boldsymbol{\eta} = 0$ or $c_{i} \equiv  \displaystyle \sum_{j = 1}^{m} g_{i j} \varphi(\eta_{j}) $ for some $i$. 
\end{lemma}
\begin{proof}
Let $\boldsymbol{F}$ be the $(m \times m)$ matrix with entries $\displaystyle \frac{\mbox{d}\varphi_{j}}{\mbox{d}\eta_{k}} =\varphi_{k} \left(\delta_{jk} - \varphi_{j}\right)$ where $\delta_{jk} = 1$ if $j = k$ and zero otherwise. Define also $\boldsymbol{e} = \left(\boldsymbol{G} \boldsymbol{\varphi} - \boldsymbol{c} \right)$ and indicate with $\mathds{I}$ a $(m \times m)$ identity matrix and with $\mathds{1}$ a $(m \times 1)$ vector of ones. Then we have
%Let $\boldsymbol{F} = \mbox{diag}\left\{\boldsymbol{\varphi}\right\}$, then the gradient and hessian functions are 
\begin{equation}\label{eq:gradHessian}
\begin{array}{llll}
L^{\prime}(\boldsymbol{\eta}) &=& \displaystyle 2 \boldsymbol{e}^{\top} \boldsymbol{G}  \boldsymbol{F}^{\top}\\
L^{\prime\prime}\left(\boldsymbol{\eta}\right) &=& \displaystyle 2 \left\{ \mbox{diag}\left( (\mathds{I} - \boldsymbol{\varphi}^{\top} \mathds{1}) \boldsymbol{G}^{\top}\boldsymbol{e}\right) +
\mbox{diag}(\boldsymbol{\varphi}) \mathds{1} \boldsymbol{e}^{\top} \boldsymbol{G} + 
\boldsymbol{F}^{\top} \boldsymbol{G}^{\top}\boldsymbol{G}\right\}\boldsymbol{F} \\
Q^{\prime}(\boldsymbol{\eta}) &=& \displaystyle  2\lambda \boldsymbol{D}^{\top} \boldsymbol{D} \boldsymbol{\eta},\\
Q^{\prime\prime}(\boldsymbol{\eta}) &=& \displaystyle 2  \lambda  \boldsymbol{D}^{\top} \boldsymbol{D}.
\end{array}
\end{equation}
The first and second elements of the Hessian $L^{\prime\prime}\left(\boldsymbol{\eta}\right)$ contain scalar products of the predictor matrix by the residual vector which are nearly orthogonal. For this reason we can use the approximation  
\begin{equation}\label{eq:approxHessian}
L^{\prime\prime}\left(\boldsymbol{\eta}\right) = 2 \boldsymbol{F}^{\top} \boldsymbol{G}^{\top}\boldsymbol{G} \boldsymbol{F},
\end{equation}
which also makes the Hessian matrix more robust for inversion \citep[see e.g.][]{Lipovetsky2009}.

\end{proof}
\begin{theorem}
$L\left(\boldsymbol{\eta}\right) + Q(\boldsymbol{\eta})$ is convex in $\mathds{R}^{m}$ and, for a given $\lambda > 0$, Eq.~\ref{eq:pen_LS} has a unique solution.
\end{theorem}

\begin{proof}
As noticed above, $L^{\prime\prime}(\boldsymbol{\eta})$ can be approximated by Eq.~\ref{eq:approxHessian} which is nonnegative definite if it exists. $Q^{\prime\prime}(\boldsymbol{\eta})$ is also nonnegative definite. It follows that $L^{\prime\prime}(\boldsymbol{\eta}) + Q^{\prime\prime}(\boldsymbol{\eta})$ is nonnegative definite and strictly convex \citep[see e.g.][]{Fleming1987}. Therefore, Eq.~\ref{eq:pen_LS} admits a unique solution for $\lambda > 0$.
%%The vector $\hat{\boldsymbol{\varphi}}$ is obtained by solving
%%\[
%%\left(\boldsymbol{G}^{\top} \boldsymbol{G} + \lambda \boldsymbol{D}^{\top} \boldsymbol{D} \right)
%%\exp(\boldsymbol{\eta}) = \boldsymbol{G}^{\top} \boldsymbol{c}
%%\]
\end{proof}

\subsection{Estimation procedure} \label{subsec:estimation_process}
The non-linear problem in Eq.~\ref{eq:pen_LS} can be solved using a scoring procedure adapting the algorithms presented in \citet{Thompson1981} and \citet{Green1984}. Indicate with $\tilde{\boldsymbol \mu}$ an approximation of the mean function (in what follows the tilde symbol will always indicate an approximation). Then, a first order expansion gives
\begin{equation*}\label{eq:mean_func} 
	\mu_{i} \approx \tilde{\mu_{i}} + \sum_{j} \frac{\partial
	\tilde{\mu}_{i}}{\partial
	\eta_{j}} \Delta \eta_{j} = \tilde{\mu_{i}} + \sum_{j} g_{ij} 
	\tilde{\varphi}_{j} \Delta \eta_{j}.
\end{equation*}
By combining this result with Eq~\ref{eq:pen_LS}, we can derive the linearized least squares criterion
\begin{equation}\label{eq:lin_crit}
\min_{\boldsymbol{\eta}} \tilde{S}(\boldsymbol{\eta}) = \left\|\boldsymbol{c} - \tilde{\boldsymbol{\mu}} - \boldsymbol{G}\tilde{\boldsymbol{F}}(\boldsymbol{\eta} - \tilde{\boldsymbol \eta})\right\|^{2} + \lambda\left\|\boldsymbol{D} \boldsymbol{\eta}\right\|^{2},
\end{equation}
where $\tilde{F}_{jk} = \tilde{\varphi}_{k} (\delta_{jk} - \tilde{\varphi}_{j})$. The coefficients $\boldsymbol{\eta}$ can then be estimated through (penalized) iterative weighted least squares, by repeatedly solving the following set of normal equations  
\begin{equation}\label{eq:IRLS}
\left(\tilde{\boldsymbol{E}}^{\top} \tilde{\boldsymbol{E}} + \lambda
\boldsymbol{D}^{\top}\boldsymbol{D}
\right) \boldsymbol{\eta} = \tilde{\boldsymbol{E}}^{\top}
\left(\boldsymbol{c} -\tilde{\boldsymbol{\mu}} + \tilde{\boldsymbol{E}}
\tilde{\boldsymbol{\eta}} \right),
\end{equation}
with $\tilde{\boldsymbol{E}} = \boldsymbol{G} \tilde{\boldsymbol{F}}$ and $\eta_{1} = 0$ ensures identifiability. Convergence is usually achieved with a modest number of iterations (less than 30 for a relative tolerance of $10^{-5}$). 

Finally, for a given value of $\lambda$, using Eq.~\ref{eq:approxHessian}, the variance-covariance matrix of $\hat{\boldsymbol{\eta}}$ can be computed as 
\begin{equation}\label{eq:Var-Cov}
\mbox{Var}(\hat{\boldsymbol{\eta}}) = \hat{\sigma}^{2} \left( \tilde{
\boldsymbol{E}}^{\top} \tilde{\boldsymbol{E}}  + \lambda \boldsymbol{D}^{\top} \boldsymbol{D}  \right)^{-1},
\end{equation}
which can be modified to accommodate optional weights (see e.g. in Section~\ref{subsec:heteroscedasticity}). The variance of $\hat{\boldsymbol{\varphi}}$ is then obtained using the delta method. 

\subsection{Heteroscedasticity} \label{subsec:heteroscedasticity}
By introducing a suitable matrix of weights $\boldsymbol{W}$, Eq.~\ref{eq:IRLS} can be modified to accommodate heteroscedastic residuals
\begin{equation*}\label{eq:wtg_norm_eq}
\left(\tilde{\boldsymbol{E}}^{\top} {\boldsymbol{W}} \tilde{\boldsymbol{E}} + 
\lambda
\boldsymbol{D}^{\top}\boldsymbol{D}
\right) \boldsymbol{\eta} = \tilde{\boldsymbol{E}}^{\top} {\boldsymbol{W}} 
\left(\boldsymbol{c} -
\tilde{\boldsymbol{\mu}} + \tilde{\boldsymbol{E}}
\tilde{\boldsymbol{\eta}} \right).
\end{equation*}
For example, since most of the variability is expected for small values of the pay-off, one could set weights equal to the inverse of the prices/strike ratio: ${\boldsymbol{W}} = \mbox{diag}\left(\tilde{\boldsymbol{\mu}}/\boldsymbol {k} \right)^{-1}$. 
Similarly, it is possible to correct for residual autocorrelation if necessary. For instance, an AR1$(\rho)$ correlation structure can be modeled taking $W_{i,j} = \rho^{|i - j|}$ (where the auto-correlation parameter is estimated using standard GLS techniques).

\subsection{Including put options} \label{sec_parity}
Put options $\boldsymbol{p}(\boldsymbol{k}, T)$ are often traded together with the homologous calls. The \textit{put-call parity} links the two contracts 
\begin{equation}\label{eq:put_call_parity}
\boldsymbol{c}(\boldsymbol{k}, T) = \boldsymbol{p}(\boldsymbol{k}, T) + s - 
\boldsymbol{k} \exp(-r \tau).
\end{equation}
Eq.~\ref{eq:put_call_parity} allows to compute the equilibrium put prices once the value of the related calls have been estimated for different strikes (or vice versa).

We propose a different strategy. The model matrix connects the observed prices to the estimated density. By definition, the state price density is unique while $\boldsymbol{G}$ can be generalized to model different classes of contracts because its entries depend on the pay-off function only. We can then introduce put prices as additional observations by augmenting the model matrix with row vectors $g^{*}_{ij} = (k_{i} - u_{j})^{+}$ (the put pay-off) and including the quoted put prices in the vector of observations. The estimation procedure remains the same. 

\subsection{Penalty parameter selection} \label{subsec_penalty_sel}
The parameter $\lambda$ can be selected by means of a suitable optimality criterion. Well-known possibilities are (generalized) cross validation, Akaike's information criterion and the Bayesian information criterion. Figure~\ref{fig_AIC_profile} shows an example based on an AIC minimization, exploring a fine grid of $\log_{10}\lambda$.

However, we propose to adopt mixed model theory \citep[see e.g.][]{lee2006, Ruppert2003, Wood2017b}. Then $\lambda = \sigma^2/\sigma_r^2$, where $\sigma^2$ is  the variance of the observation noise and $\sigma_r^2$ the variance of the random effects, in the present case $\boldsymbol{D} \boldsymbol{\eta}$. To estimate these variances, an iterative EM procedure is effective \citep[][]{Schall1991}. It uses the updating equations  
\[
\displaystyle \hat{\sigma}^{2} = \displaystyle \frac{\displaystyle \sum_{i=1}^{n} w_{i} \left(c_{i} - \sum_{j=1}^{m} g_{ij} \hat{{\varphi}}_{j}\right)^{2}}{n - \mbox{ED}},\
\displaystyle \hat\sigma_r^{2} = \frac{\displaystyle \left\|\boldsymbol{D} \hat{\boldsymbol{\eta}}\right\|^2}{\mbox{ED} - d },
\]
where ED is the effective model dimension and $d$ the order of the differences in the penalty. Following \cite{Hastie1990}, ED is defined as the trace of the hat matrix $\boldsymbol{H}$, with
$$  
\boldsymbol{H} = 
\tilde{\boldsymbol{E}}( \tilde{\boldsymbol{E}}^{\top} \boldsymbol{W}
\tilde{\boldsymbol{E}} + \lambda \boldsymbol{D}^{\top} \boldsymbol{D}  )^{-1} 
\tilde{\boldsymbol{E}}^{\top}\boldsymbol{W}.
$$ 
Convergence is achieved in a limited number of iterations (usually less than 15). Under the normality assumption, a strict relationship connects AIC optimization and the mixed model approach, since the $\lambda$ selected by the latter method can be shown to be equal to the expected value of the one suggested by the AIC criterion \citep[see][]{Krivobokova2007}. 

\begin{figure}
\centering
\includegraphics[width = 1\linewidth]{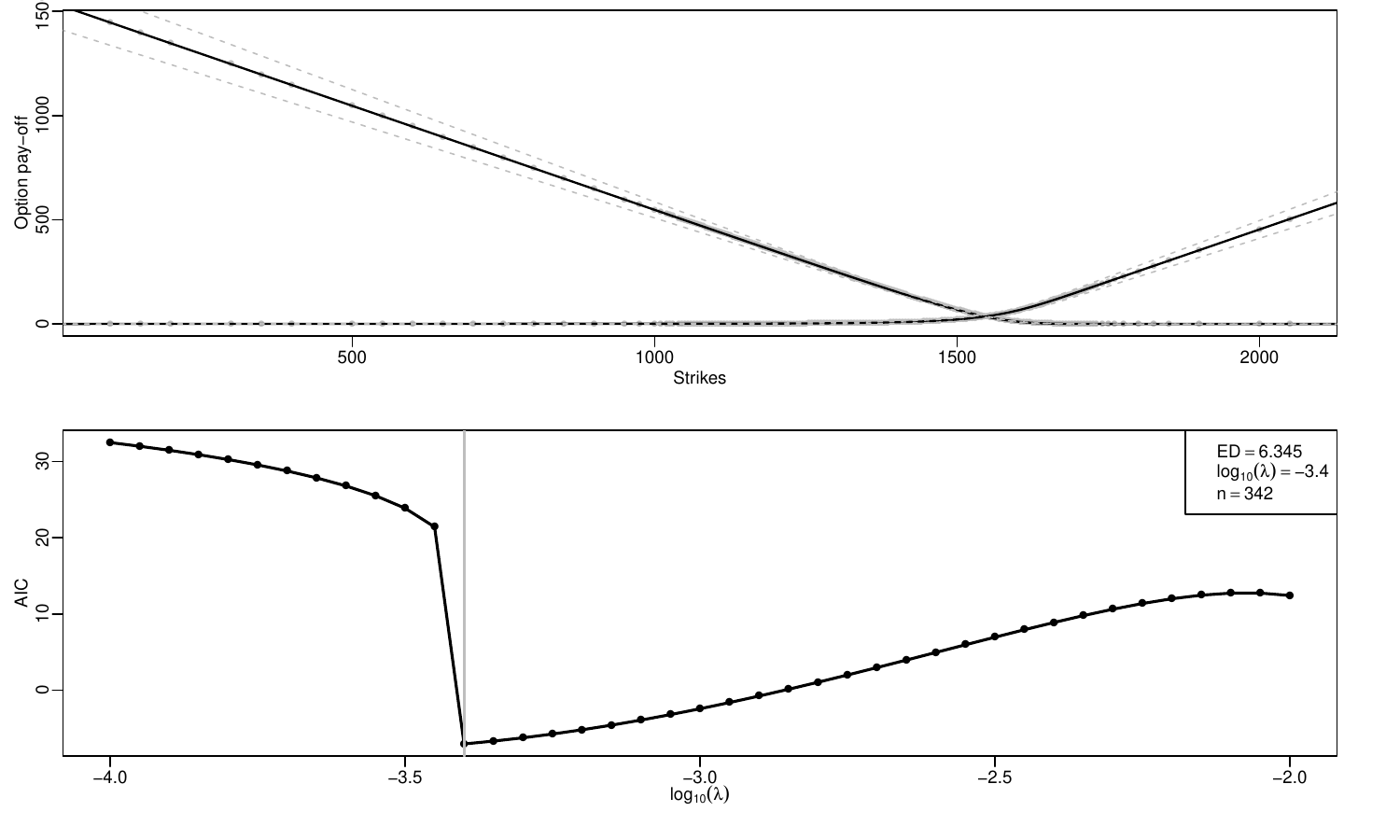}
\caption{Upper panel: observed (171) put and (171) call option prices (gray dots), smooth pay-offs (solid black lines) and point-wise confidence bounds 
(gray dashed lines). Lower panel: AIC values for different values 
of $\log_{10}(\lambda)$. The vertical line indicates the minimum 
of the selection criterion (also reported in the legend with the associated effective dimension). The data are available in the R-package \texttt{RND} (S\&P 500 options with 62 days to maturity quoted on April 19, 2013).}
\label{fig_AIC_profile}
\end{figure}

\subsection{Definition of the $\boldsymbol{u}$-set}
\label{subsec:gridvalue}
We only require for the set of possible underlying prices at maturity (vector $\boldsymbol{u}$ in the model matrix definition) to satisfy $u_{1} < k_{1}$ and $u_{m} > k_{n}$ (a good choice in most applications is $u_{1} = \max(0, k_{1} - 0.1 \times k_{1})$ and $u_{m} = k_{n} + 0.1 \times k_{n}$). Defining this vector is equivalent to choosing the numerical integration scheme for Eq.~\ref{eq:integral_price}. We have tested two strategies: 1) uniform and equally spaced set of $u$ values (as in Section~\ref{sec_settings_estimation}) equivalent to a simple rectangular integral approximation, 2) a Gaussian quadrature scheme. This second strategy can be easily accommodated in what we presented in Section~\ref{sec_settings_estimation}. The penalized least squares problem becomes
\begin{equation}\label{eq:pen_LS_quadrature}
\min_{\boldsymbol{\eta}} S(\boldsymbol{\eta}) = \sum^{n}_{i=1} \left(c_{i} - \sum^{m}_{j=1} g_{ij} w_{j} \varphi_{j}(\boldsymbol{\eta}) \right)^{2} + \lambda \|\boldsymbol{D} \boldsymbol{\eta}\|^{2},
\end{equation}
where $g_{ij}$ is now the $i,j$th element of the composition matrix defined in Section~\ref{sec_settings_estimation} but with vector $\boldsymbol{u}$ corresponding to the integration nodes, $w_{j}$ is the $j$th integration weight associated to the node $u_{j}$, $\varphi_{j}(\boldsymbol{\eta}) = \displaystyle \frac{\exp(\boldsymbol{B}(u_{j}) \boldsymbol{\eta})}{1 + \sum_{k = 2}^{m} \exp(\boldsymbol{B}(u_{k}) \boldsymbol{\eta})}$ with $\boldsymbol{B}$ a B-spline basis evaluated at $\boldsymbol{u}$ and $\boldsymbol{\eta}$ a set of spline coefficients to be estimated. We suggest to build the bases using a generous number of equally spaced knots \citep[see e.g.][]{Eilers2010}.

In our experience, the two strategies ensure very similar results. In Figure~\ref{fig_unifovm_vs_Gauss} we compare the price fitting accuracy of the uniform and Gauss-Legendre quadrature approaches using real data.    
\begin{figure}
\centering
\includegraphics[width = 1\linewidth]{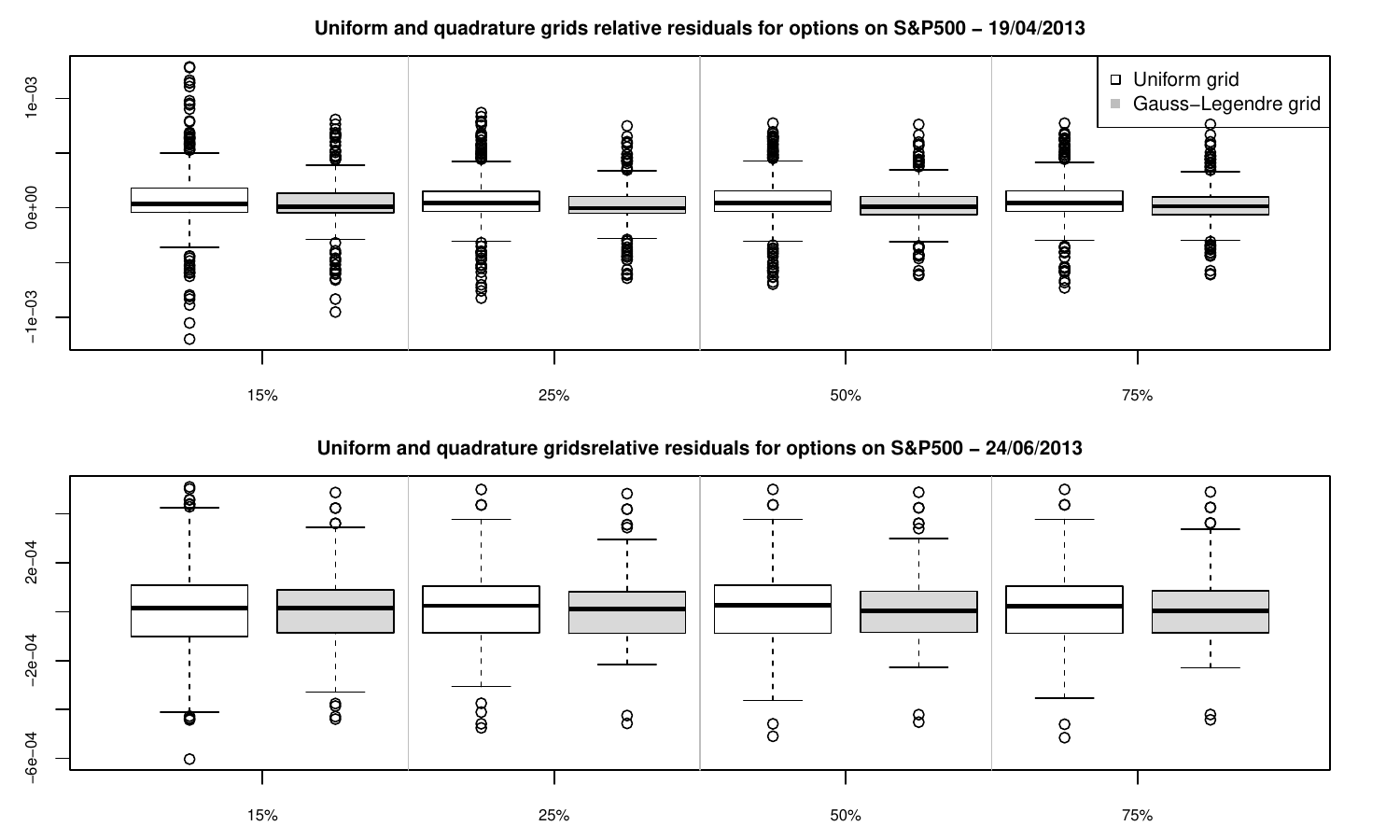}
\caption{Relative (to the underlying prices) residuals observed for uniform grids and  Gauss-Legendre integration scheme (B-splines defined over 40 equally spaced internal knots). We compare grids of $\{15\%, 25\%, 50\%, 75\% \}$ the sample size $n$. The data are available in the R-package \texttt{RND}. The relative (to the underlying price) differences between the results of the two integration strategies never exceed 0.5\%.}
\label{fig_unifovm_vs_Gauss}
\end{figure}

\section{No-arbitrage properties of the estimates}\label{sec_noarbitrage}
\begin{theorem}
The estimated option prices satisfy a set of no-arbitrage conditions (for brevity, the list below refers to call contracts only, \citealp[see e.g.][]{Harrison1981, Hardle2009}):

\begin{itemize}\itemsep1pt
  \item[1.] The estimated density is proper: $\hat{\varphi}_{j} \geq 0,\ \forall j=1,...,m$ and $\displaystyle \sum_{j = 1}^{m} \hat{\varphi}_{j} = 1$,
  \item[2.] The estimated prices are non-negative: $\hat{{c}}_{i} \geq 0,\ \forall i = 1,...,n$,
  \item[3.] The pricing function is monotone in the strikes:  $\displaystyle \frac{\partial \hat{{c}}_{i}}{\partial k_{i}} \leq 0,\ \forall i = 1,...,n$,
  \item[4.] The pricing function is convex in the strikes: $\displaystyle \frac{\partial^{2} \hat{{c}}_{i}}{\partial k^{2}_{i} } \geq 0,\ \forall i = 1,...,n$.
\end{itemize}
These conditions hold by definition for the DESPD estimator and do not need to be explicitly imposed in the estimation process.
\end{theorem}

\begin{proof}
Property 1 is satisfied since, by definition, 
$\hat{\varphi}_{j} = \displaystyle \frac{\exp(\hat{\eta}_{j})}{1+\displaystyle \sum_{k = 2}^{m} \exp(\hat{\eta}_{k})} \geq 0,\ \forall\ \hat{\eta}_{j} \in \mathds{R}^{m}$ and $\displaystyle \sum_{j = 1}^{m} \hat{\varphi_{j}}= 1$. This also proves condition 2 given the definition of the model matrix $\boldsymbol{G}$. In order to prove conditions 3 and 4 it is convenient to express the $i$th estimated call price as
	\[
	\hat{c}_{i} = \sum_{j = 1}^{m}\left[u_{j} - k_{i}\right]\mathtt{H}\left(u_{j} - k_{i}\right)\hat{\varphi}_{j},
	\]
where $\mathtt{H}\left(u_{j} - k_{i}\right)$ is a Heaviside step function with value one if $u_{j} \geq k_{i}$ and zero otherwise.  Define now the delta function $\delta \left(u_{j} - k_{i}\right)$ equal to zero for $u_{j} \neq k_{i}$. Computing the first and second derivatives (w.r.t. $k_{i}$) of the estimated option prices we obtain
	\[
	\begin{array}{llll}
	\displaystyle
	\displaystyle \frac{\partial \hat{c}_{i}}{\partial k_{i}} &= \displaystyle - \sum_{j = 1}^{m} \mathtt{H}\left(u_{j} - k_{i}\right) \hat{\varphi}_{j} - \sum_{j = 1}^{m} \left[u_{j} - k_{i}\right] \delta \left(u_{j} - k_{i}\right)\hat{\varphi}_{j} = \displaystyle - \sum_{j = 1}^{m} \mathtt{H}\left(u_{j} - k_{i}\right) \hat{\varphi}_{j}\leq 0,\\
	\displaystyle \frac{\partial^{2} \hat{c}_{i}}{\partial k^{2}_{i}} &= \displaystyle  \sum_{j = 1}^{m} \delta \left(u_{j} - k_{i}\right) \hat{\varphi}_{j} \geq 0. %\approx \hat{\varphi}(k_{i})
	\end{array}
	\]
By condition 1 and the definition of $\mathtt{H}\left(u_{j} - k_{i}\right)$, it also follows $\displaystyle \frac{\partial \hat{c}_{i}}{\partial k_{i}} \in \left[-1, 0 \right]$; note that this is also a necessary condition for arbitrage-free estimates, see e.g. \citealp*{Sahalia2003}. Conditions 1 and 2 clearly hold as well for put prices. Conditions 3 (with a positive sign now) and 4 can be proved analogously to call prices since
\[
\hat{p}_{i} = \sum_{j = 1}^{m}\left[k_{i} - u_{j}\right]\mathtt{H}\left(k_{i} - u_{j}\right)\hat{\varphi}_{j}.
\]
\end{proof}

Finally, it is required that the mean of the SPD is equal to the forward price.
This condition can be imposed by appropriately shifting the estimation support \citep[a similar strategy is discussed in][]{Hardle2009}. 
Once the vector $\hat{\boldsymbol{\varphi}}$ has been obtained we define the adjusted $\tilde{\boldsymbol{u}}$ vector as follows:
\[
 \displaystyle\tilde{\boldsymbol{u}} =  \displaystyle \boldsymbol{u} - \mathtt{I}_{\{+\}}^{(u)} \times \left(\boldsymbol{u}^{\top}\hat{\boldsymbol{\varphi}} - \exp(r(t-\tau))s_{t} \right) +
     \displaystyle \mathtt{I}_{\{-\}}^{(u)} \times \left(\exp(r(t-\tau))s_{t} - \boldsymbol{u}^{\top}\hat{\boldsymbol{\varphi}} \right), 
\]
where $\mathtt{I}_{\{-\}}^{(u)}$ is an indicator function equal to 1 if $\left(\boldsymbol{u}^{\top}\hat{\boldsymbol{\varphi}} < \exp(r(t-\tau))s_{t} \right)$ (and zero otherwise) and $\mathtt{I}_{\{+\}}^{(u)}$ is equal to 1 if $\left(\boldsymbol{u}^{\top}\hat{\boldsymbol{\varphi}} > \exp(r(t-\tau))s_{t} \right)$.

\section{Simulation analysis} 
\label{sec_simulations}
In this section we test our approach through a simulation study. We compare the DESPD model with the positive convolution approximation \citep[PCA,][]{Bondarenko2003}, the implied volatility smoothing of \citet{Fengler2009} and the nonlinear least squares procedure by \citet{Hardle2009}. The method of \citet[][]{Bondarenko2003} has already been found superior to the following alternatives:
%\begin{itemize}
i) a mixture of log-normal densities \citep[see e.g.][]{Bahra1996, Melick1997},
ii) Hermite polynomials/Edgeworth expansion \citep[see e.g.][]{Abken1996, Jarrow1989},
iii) regularization methods \citep[specifically][proposal]{Jackwerth1996},
iv) regularization of the implied volatility \citep*[see e.g.][]{Campa1998, Bliss2002, Jackwerth2000},
v) Modeling the SPD with sigma shaped polynomials \citep{Rosenberg1998}. We doe not compare them directly to our proposal.
%\end{itemize}

Our evaluations are based on 1000 simulated European put options with maturity 21 days  \citep[in analogy with][]{Bondarenko2003}. We define sets of 7 and 23 regularly spaced strike prices $k_{i} = 430, 435, \dots, 540$ and assume that the underlying follows a mixture of three log-normal ($\mathcal{LN}$) densities,
	\[
	f^{*}(k) = \alpha_{1}\mathcal{LN}\left(k| \eta_{1}, \zeta_{1} \right) +
		 	    \alpha_{2}\mathcal{LN}\left(k| \eta_{2}, \zeta_{2} \right)+
		    	\alpha_{3}\mathcal{LN}\left(k| \eta_{3}, \zeta_{3} \right)
	\]
with parameters as in Table~\ref{tabLNpar}.

\begin{table}[htbp]
  \centering
  \caption{Simulation settings: parameters of the log-normal mixture}
    \begin{tabular}{lll}
    \hline
    $\alpha_{1} = 0.1194$ & $\eta_{1} = 475.59$ & $\zeta_{1} = 0.0550$ \\
    $\alpha_{2} = 0.8505$ & $\eta_{2} = 498.17$ & $\zeta_{2} = 0.0206$ \\
    $\alpha_{3} = 0.0301$ & $\eta_{3} = 524.91$ & $\zeta_{3} = 0.0146$ \\
    \hline
    \end{tabular}%
  \label{tabLNpar}%
\end{table}%

The put prices $\boldsymbol{p}^{*}$ have been contaminated by zero-mean uniform additive random noise  proportional to the bid-ask spread $\mathcal{S}$ (details can be found in \citealp*{Bondarenko2003}, Appendix B.2) with two levels of variability: $\epsilon \sim \mathcal{U}\left(-0.5\mathcal{S}, 0.5\mathcal{S} \right)$ and $0.5\epsilon$.

Our assessments are based on three performance indicators:
\begin{enumerate}
  \item  The goodness of fit of the recovered option prices to the simulated (noise-free) ones, as measured by the root Mean Squared Errors (RMSE) between the theoretical put prices ($\boldsymbol{p}^{*}$) and the estimated ones ($\hat{\boldsymbol{p}}$)
	\[
	\mbox{RMSE} = \sqrt{\sum_{i=1}^{n} \frac{\left(p_{i}^{*} - \hat{p}_{i}\right)^{2}}{n}}.
	\]
  \item The quality of the fit, as measured by the Root Integrated Squared Error (RISE) between theoretical ($f^{*}$) and estimated ($\hat{\varphi}$) SPDs
	 \begin{equation*}
	 \mbox{RISE} =\frac{1}{ \|f^{*}(x) \|} \sqrt{\int_{0}^{\infty}{ \left( f^{*}(x) - \hat{\varphi}(x) \right)^{2} \mbox{d} x} }.
	\end{equation*}
\item The accuracy of the inferred risk neutral moments measured by the absolute deviations between the estimated and the theoretical first four moments.	
\end{enumerate}

For each simulated data set, the optimal PCA bandwidth was selected by minimizing the root mean integrated error w.r.t. the simulated risk-neutral density (the same strategy is adopted in \citealp{Bondarenko2003}). However, this criterion is not feasible with real data when the SPD is unknown. In contrast, the DESPD results are based on a grid of 200 equally spaced $u$-points and the smoothing parameters were optimized as described in Section~\ref{subsec_penalty_sel}. The estimates were obtained using the results of Section~\ref{subsec:estimation_process} (the iterative weighted least squares reached convergence in less than 25 iteration on average for each simulation run).

Figures~\ref{fig_simulation_23k} and~\ref{fig_simulation_7k} summarize our results. 
The RMSE (middle upper and lower panels) between the theoretical and estimated prices, are moderate for the four competing methods (for both sample sizes). This highlights the good price fitting properties of the procedures. However, DESPD and PCA seem more accurate than the other two methods (this is evident in the simulation with 23 strikes). The root integrated squared errors between the estimated and the theoretical SPDs (leftmost upper and lower panels) quantify the accuracy achieved by the four estimators in estimating the latent distribution. In this case, DESPD shows excellent performance for the two sample sizes across the different noise levels. Finally, DESPD ensures the most accurate estimates of the (first four) risk neutral moments across all sample sizes and noise levels (rightmost upper and lower panels).

\begin{figure}
\centering
\includegraphics[width = 1\linewidth]{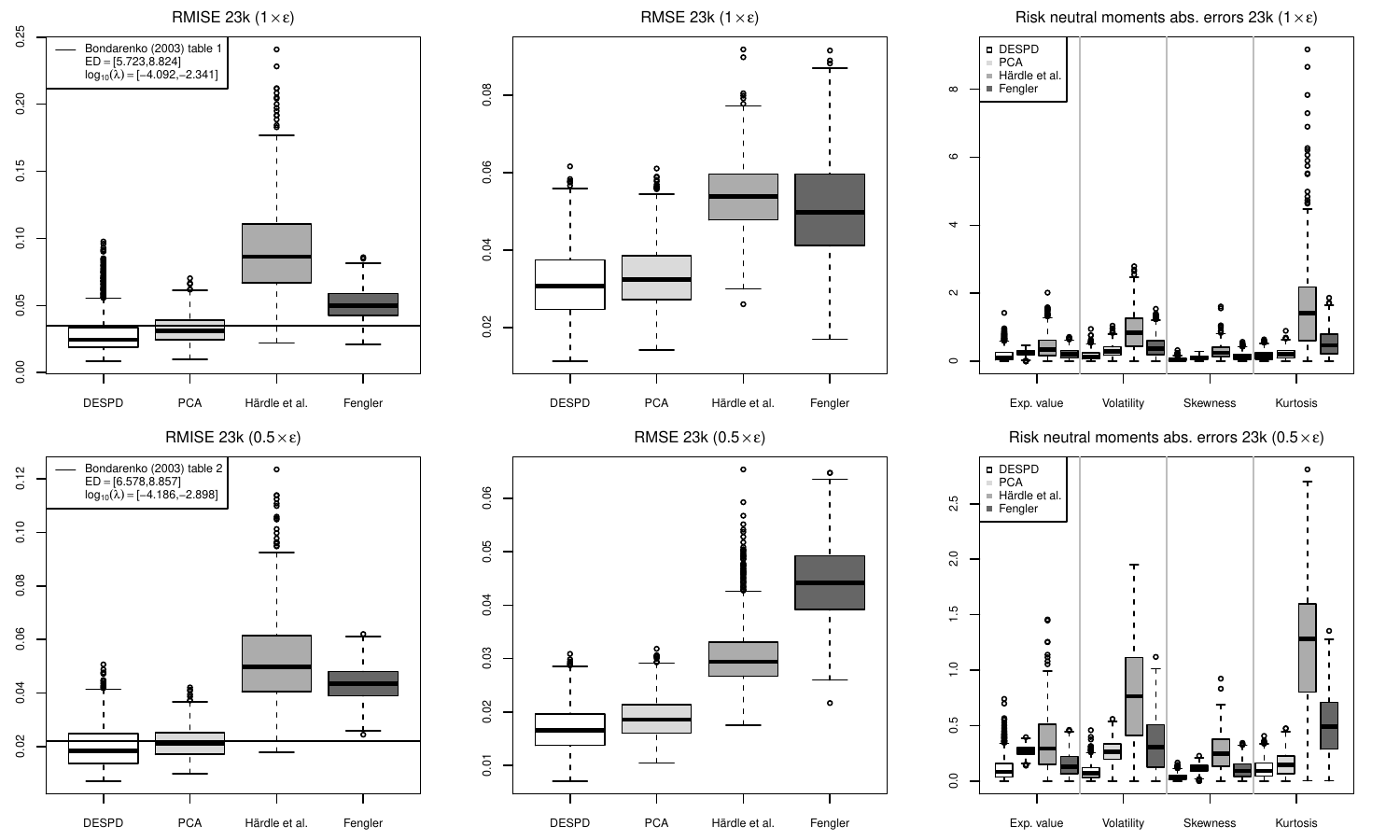}
\caption{Boxplots of RMISEs, RMSEs and absolute (first four) moments estimation errors computed over the simulations for the four estimators. These results refer to the prices simulated for 23 strikes. In the legends of the RMISE plots we indicate the ranges observed for the effective model dimension and the optimal $\log_{10}(\lambda)$ (see Section~\ref{subsec_penalty_sel}).}
\label{fig_simulation_23k}
\end{figure}
\begin{figure}
\centering
\includegraphics[width = 1\linewidth]{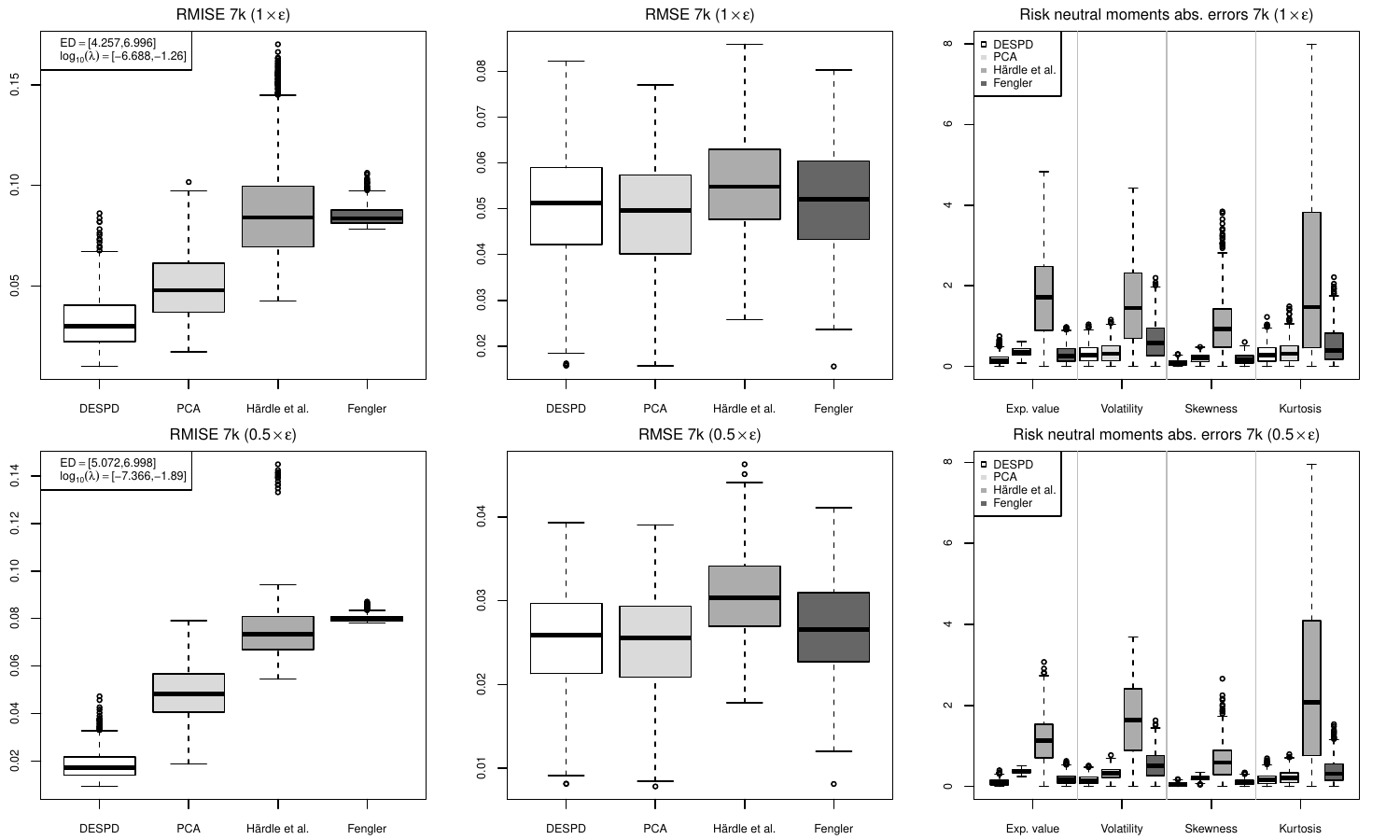}
\caption{Boxplots of RMISEs, RMSEs and absolute (first four) moments estimation errors computed over the simulations for the four estimators. These results refer to the prices simulated for 7 strikes. In the legends of the RMISE plots we indicate the ranges observed for the effective model dimension and the optimal $\log_{10}(\lambda)$ (see Section~\ref{subsec_penalty_sel}). }
\label{fig_simulation_7k}
\end{figure}

\subsection{Extrapolation and tail behavior}
\label{sec_extr_tail}
Over unobserved support points, DESPD density extrapolation is driven by the penalty term. The third order penalty implies an extrapolating polynomial of degree 2 for the logarithm of the density \citep*{Eilers2010}. 

In order to investigate extrapolation performance over unobserved support points, we produced a second simulation experiment. 
We generate two sets of options prices following a mixture of log-normal distributions with parameters as in Table~\ref{tabLNpar}. In analogy with Section~\ref{sec_simulations} we sample option prices for 23 strikes with two levels of noise. During the estimation process, we treat the strikes $k \in \{430, \dots, 460\}$ and $k \in \{525, \cdots, 540\}$ as unobserved. 

Our results are presented in Figure~\ref{fig_tails} which also includes the estimates obtained by extrapolating the tails of the inferred distributions using the parametric method in \citet{Figlewski2009}. For the optimization of the extreme value density parameters we follow \citet{Birru2012}. The DESPD estimates (gray and red lines) properly approximate the theoretical distribution (black dashed line). This holds across the observed and unobserved domain for both left and right tails. On the other hand, the right tails sampled from the optimized GEV densities (blue lines) tend to underestimate the real one while the approximation of the left tail seems slightly more appropriate.
\begin{figure}
\centering
\includegraphics[width = 1\linewidth]{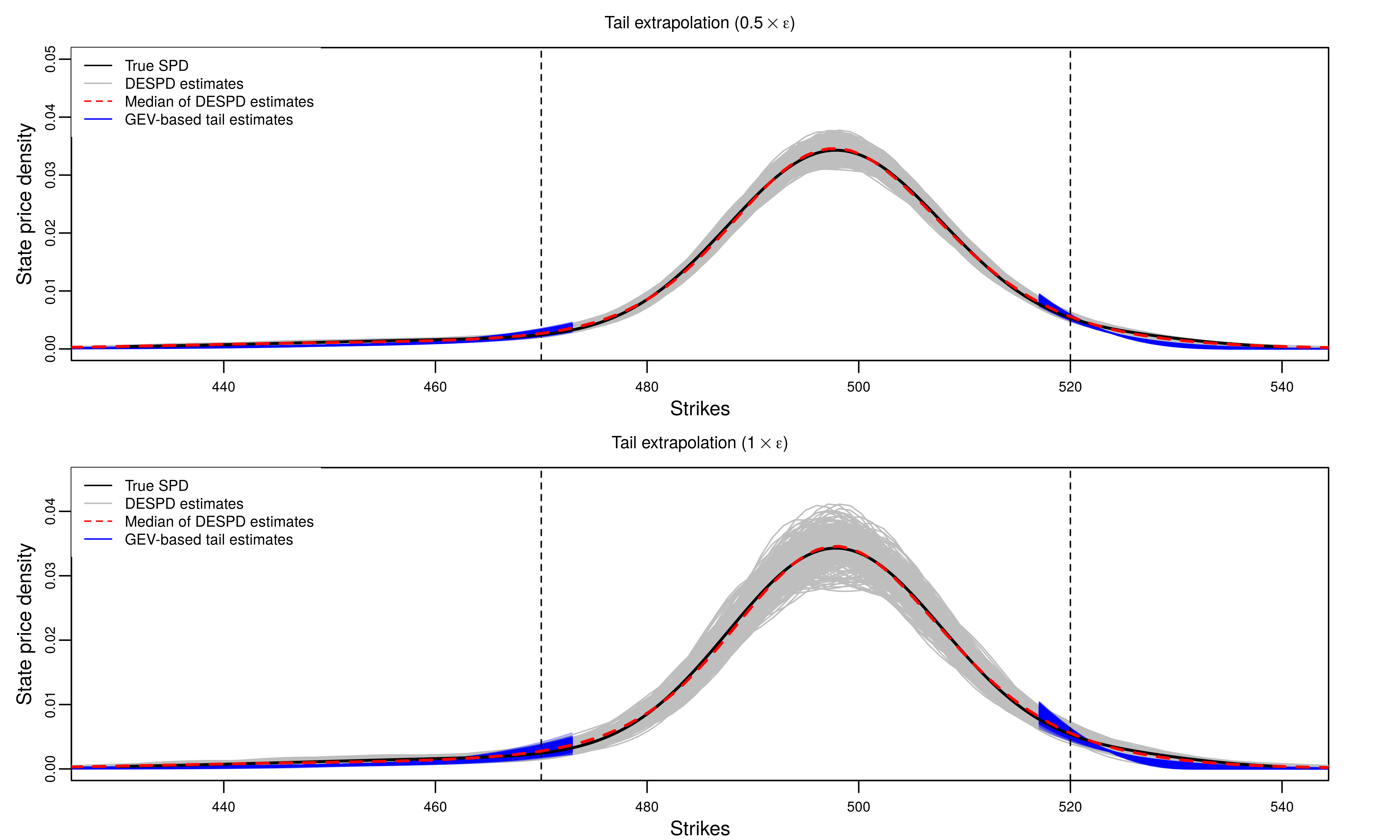}
\caption{Estimated and simulated SPDs for two different levels of price noise (23 strike prices).
The solid black line indicates the theoretical density, the dashed red lines
the median estimates along the simulation experiment and the solid gray lines the recovered densities obtained for each synthetic option prices set. The blue lines on the right and left tails represent the GEV-based estimates. For this simulation we considered the prices associated with $k \in \{430, \dots, 460\}$ and $k \in \{525, \cdots, 540\}$ as unobserved.}
\label{fig_tails}
\end{figure}

\section{Real data analysis} \label{sec_application}
We analyze a set of weekly S\&P500 call and put options (a.k.a. SPXW) traded between 02-Jan-2018 and 15-Feb-2019. SPXW options have been introduced in 2005 to provide a wider gamma of contracts for hedging and speculative purposes. These European-style contracts are settled at closure of their last trading day, which may be on a Monday, Wednesday or Friday. At a given date, the exchange lists six contracts with consecutive maturities (current expiration excluded). In this section we will focus on contracts expiring on Friday (also referred as end-of-week, EoW) only.

Figure~\ref{fig_3_mats} shows the estimates obtained for three contracts quoted on 01-Feb-2018 and expiring in one, two and three weeks. The inferred densities for each maturity appear left-skewed. The distribution of the standardized residuals (corrected for serial correlation as in Section~\ref{subsec:heteroscedasticity}) is bell-shaped and centered around 0 even if we observe slightly heavy tails especially for larger maturities (gray bars in the rightmost panels of Figure~\ref{fig_3_mats}). On the other hand, the distribution of the ``raw'' residuals (red bars in the rightmost panels of Figure~\ref{fig_3_mats}) appears skewed and with heavier tails. Finally, as expected, the volatility implied by the inferred SPD increases with the maturity: standard deviations are equal to 77.05, 79.47 and 87.67 respectively for the three contracts. 
\begin{figure}
\centering
\includegraphics[width = 1\linewidth]{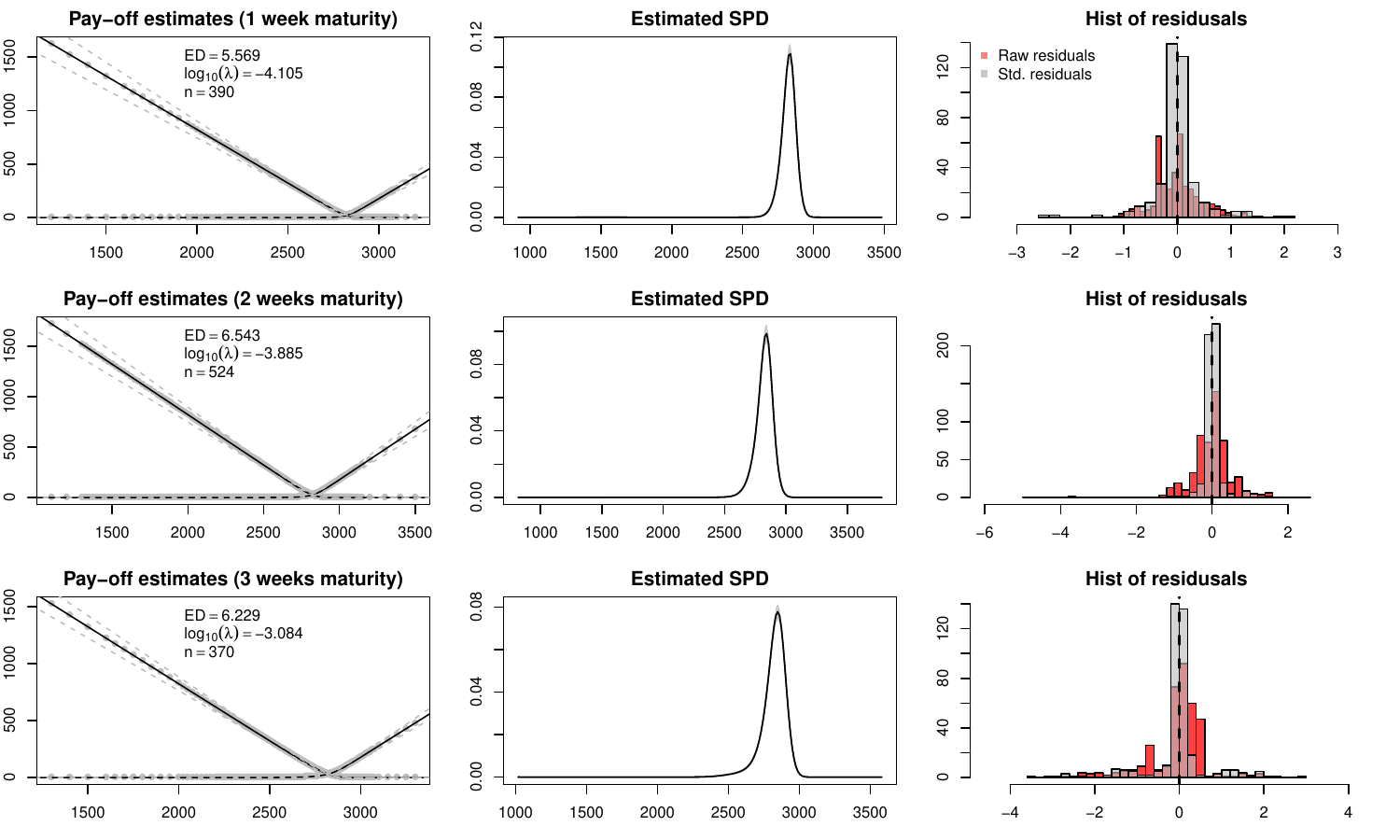}
\caption{DESPD estimates for the S\&P500 weekly option prices with maturity 1, 2 and 3
weeks. The option prices (gray dots) have been observed on February 1, 2018.}
\label{fig_3_mats}
\end{figure}

\subsection{Moments of the estimated SPD over time} 
 \label{subsec_moments}
The upper and lower left panels of Fig.~\ref{fig_spddyn} show the means and standard deviations of the SPDs estimated for the EoW-SPXWs quoted between 02-Jan-2018 and 15-Feb-2019. At each date, we used the contract with shortest maturity only. 

The first plot (top left) compares the forward underlying prices (dots) with the expected ones (lines). The expected values under the recovered SPDs follow really closely the index values. This is consistent with no-arbitrage requirements (see Sec.~\ref{sec_noarbitrage}). The solid black lines in the same graph indicate 5th and 95th percentiles under the inferred risk-neutral measures. 

The results presented in the bottom left panel summarize the standard deviations estimated at each date under the inferred SPD. As expected, the uncertainty about the value of the underlying asset at maturity decreases (almost linearly) when the expiration date (indicated by vertical dotted lines) approaches. \citet{Hardle2009} observed a similar behavior for DAX monthly options.  

The top plot in the right column of Fig.~\ref{fig_spddyn} displays the observed (gray dots) and expected (under the estimated SPD) index log-returns with standard deviations bands. The second order return moments were computed using a 5-days (one trading week) rolling window. 

The lower plot of the right column of the same figure compares the observed Cboe-VIX index (gray line) with the same index computed under the inferred risk-neutral density (black line). The VIX is a risk-neutral forward measure of market volatility calculated as the fair value of a 30-days variance swap rate (the index is computed using the model-free methodology to price a variance swap via static replication, see e.g. \citealp*[][]{Barletta2019}). The DESPD estimates were obtained by substituting the recovered implied volatilities in the formula described in~\cite{CBOE2019} technical report (page 9). The two series appear highly correlated with Pearson's $\rho$ equal to 0.987 (95\% CI: $[0.983, 0.989]$). However, the Cboe-VIX usually tends to  exhibit higher values (this is particularly evident around the spikes observed in February 2018 for example).
\begin{figure}
\centering
\includegraphics[width = 1\linewidth]{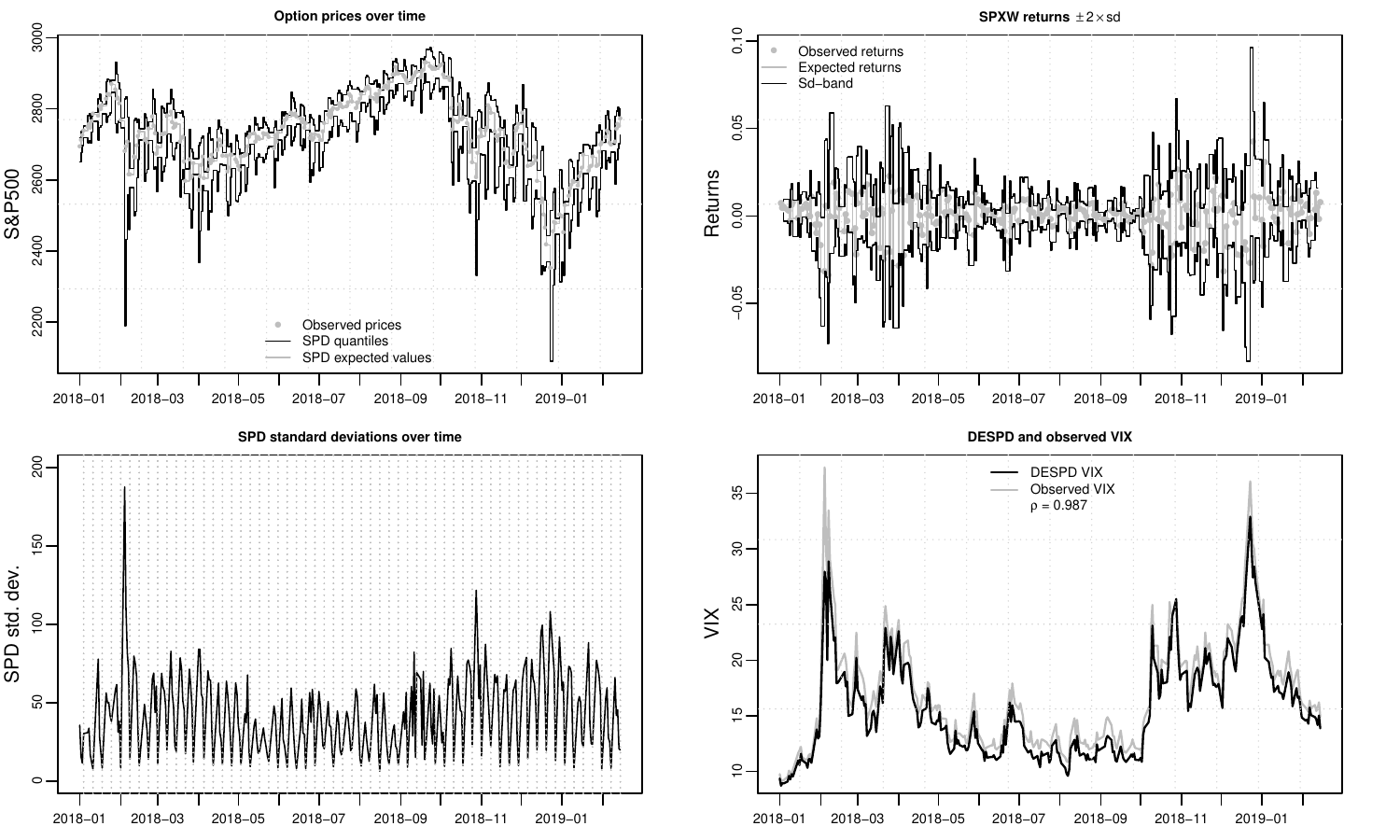}
\caption{Moments of the SPDs estimated for EoW-SPX options between 02-Jan-2018 and 15-Feb-2019. The upper panels report the time series of the index prices with expected (under the estimated densities) values and 95-percentile bands and the evolution of the implied standard deviations over time. The lower left graph summarizes the observed and expected index returns with standard deviation band. The lower right panel compares the observed VIX index with the one computed under the inferred risk neutral densities.}
\label{fig_spddyn}
\end{figure}

\subsection{Market expectations and forecast assessment}\label{sec_price_dyn}
The state price density summarizes the market expectations about the dynamics of the underlying asset under the risk-neutral probability measure. It is then interesting to evaluate if our estimates correctly synthesize such expectations and how those compare to the observed dynamics. 

We consider again end-of-week SPXW contracts. We estimate the SPD at first quotation day (58 contracts between 02-Jan-2018 and 15-Feb-2019). This gives us a one week horizon for assessing market's expectations about the movements of the S\&P500 index. From the densities inferred on each Thursday (issue date for contract expiring on next Friday), we compute the 5th and 95th percentiles and compare them with the realized index values. The upper panel of Figure~\ref{fig_mktExp} summarizes our results. The observed underlying dynamics appear consistent with our expectations. The S\&P500 prices stay within the corridor of the estimated order statistics. Only on 5 days the index reached values slightly above the 95th percentile. 
\begin{figure}
\centering
\includegraphics[width = 1\linewidth]{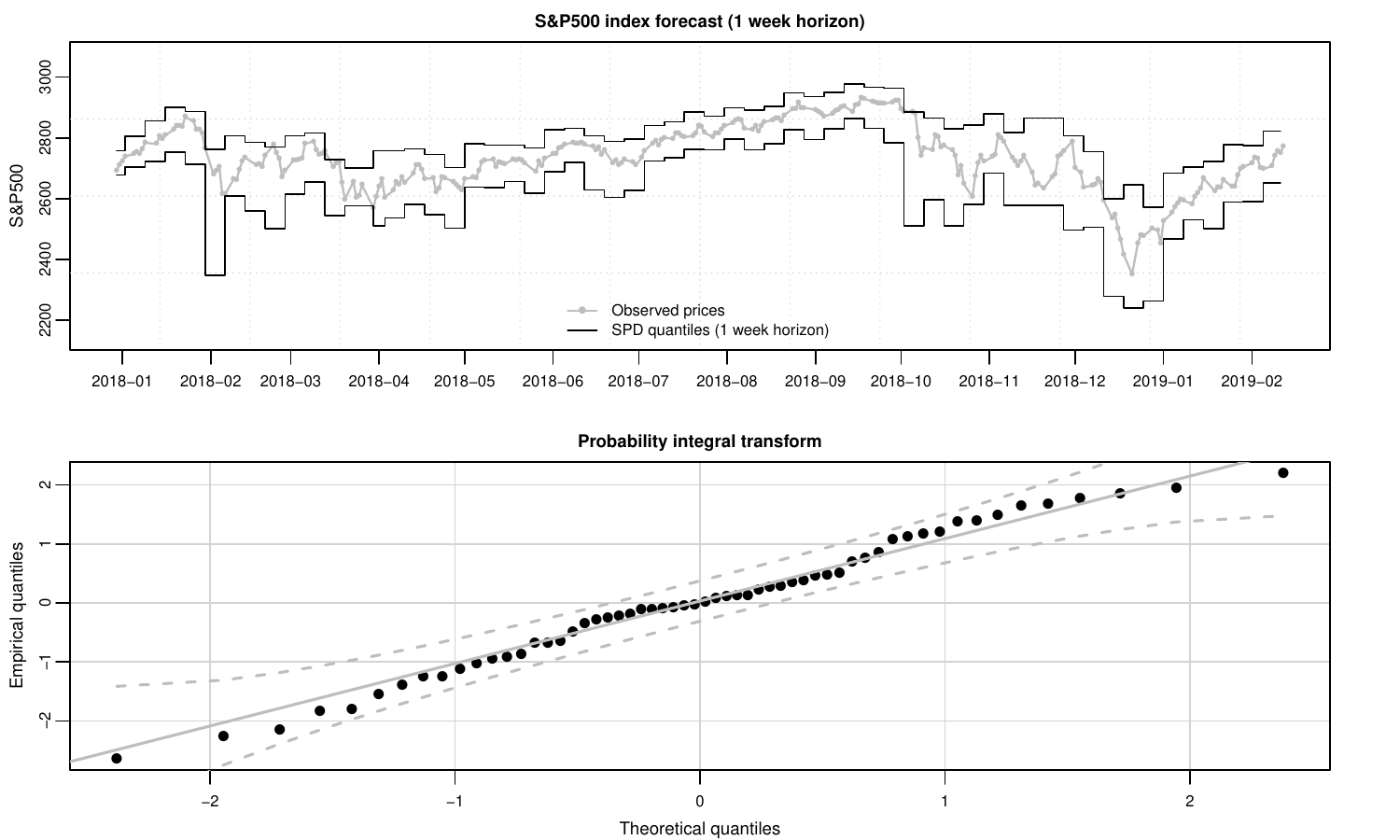}
\caption{Upper panel: one week horizon market expectations under the estimated SPDs and realized prices. The black lines indicate 5th and 95th percentiles of the risk neutral measure computed on Thursday of every week (58 quotation days). The gray line indicate the observed S\&P500 prices. The realizations laying outside the percentile range are indicated by square black dots. Lower panel: quantile-quantile plot of normalized probability integral transforms for the EoW-SPXW option contracts observed between 02-Jan-2018 and 15-Feb-2019. In the qq-plots, the empirical quantiles (black dots) are compared with the theoretical normal ones (solid  line). The confidence envelopes (dashed lines) are based on the standard errors of the order statistics of an independent random sample drawn from a standard normal distribution.}
\label{fig_mktExp}
\end{figure}

The quality of the SPD-based forecasts can be measured through probability integral transforms (PIT). Suppose that at time $t$ the density function $\hat{\boldsymbol{\varphi}}_{t}$ is available. Suppose also that we observe the spot price at maturity (end-of-week in our case) $s_{t+\Delta t}$. Then we can compute
\begin{equation*}
z_{t+\Delta t} = \int_{-\infty}^{s_{t+\Delta t}} \hat{\boldsymbol{\varphi}}_{t}(x) 
\mbox{d}x = \hat{\mathcal F}_{s}(s_{t+\Delta t}),
\end{equation*}
where $\mathcal{F}(\cdot)$ indicates the cumulative distribution function and $z_{t + \Delta t}$ is equal to the probability of $s_{t+\Delta t}$ under the density $\hat{\boldsymbol{\varphi}}_{t}$. If the estimated SPD well approximates the true density of $s_{t+\Delta t}$  then $z_{t + \Delta t} \sim U(0, 1)$ must hold, since
\begin{eqnarray*}
\mathcal{F}_{z}({z_{t+\Delta t}}) &=& \Pr(z \leq {z_{t+\Delta t}}) 
= \Pr\left(\hat{\mathcal F}_{s}({s_{t+\Delta t}}) \leq  z_{t+\Delta t} 
\right) 
= \Pr\left({s_{t+\Delta t}} \leq \hat{\mathcal F}_{s}^{-1}({z_{t+\Delta t}})
\right) \\
&=& \hat{\mathcal F}_{s}\left(\hat{\cal F}_{s}^{-1}({z_{t+\Delta t}}) \right)
= {z_{t+\Delta t}}
\end{eqnarray*}
We prefer to normalize the PIT measure \citep[see e.g.][]{Berkowitz2001} as $x_{t + \Delta t} = \Phi^{-1} (z_{t + \Delta t})$ (with $\Phi(\cdot)$ the standard Gaussian c.d.f.). Our assessments in Figure~\ref{fig_mktExp} (lower panel) confirm the appropriateness of the estimated SPDs in describing the expected underlying index dynamics at maturity. The empirical PIT quantiles are consistent with those of a standard normal distribution.  

As a final remark, we stress the fact that these results have been obtained under a strong assumption about the pricing kernel linking the risk-neutral and physical probability measures of the index realizations. The short maturity of the contracts we dealt with makes this assumption reasonable. Further support comes from the consistency between the inferred and realized price dynamics across a large number of contracts.

\section{Discussion} \label{sec_discussion}
We have introduced a new semi-parametric method, DESPD, for the direct estimation of the state price density (SPD) implied in option prices. Our framework is flexible and accurate, does not rely on any parametric assumption about the dynamics of the underlying asset and provides arbitrage-free pricing functions. It is a special case of the penalized composite link model \citep{Eilers2007} and is estimated through an intuitive iterative weighted least squares procedure.

The adjective ``direct" stresses the fact that we simultaneously estimate the unknown probability measure and the pricing function. This is different from many smoothing-based alternatives which approximate the observed prices as a smooth function and derive the unknown density by (numerical) differentiation (see e.g. \citealp*{Sahalia1998,Sahalia2000, Sahalia2003,Shimko1993,Bliss2002, Yatchew2006, Fengler2009}). 

Using P-splines, we model the logarithm of the SPD as a smooth function and match, in the least squares sense,  the expected values of the possible pay-offs at maturity with the observed prices. This guarantees the non-negativeness of the recovered density with no need for an explicit constraint or additional adjustments \citep[in contrast to, e.g.][]{Jackwerth1996}. 

The SPD estimation problem is ill-posed. Therefore, we regularize it by using a roughness penalty on the coefficients defining the unknown density. This ensures smooth estimates and efficient extrapolation to unobserved pay-offs (see results in Section~\ref{sec_extr_tail}). The optimal degree of smoothing can be selected through well established procedures (such as AIC, BIC or cross validation), but we use an efficient EM approach, inspired by mixed model theory \citep{Schall1991, Wood2017b}. In contrast, smoothness is not guaranteed by the approach of \citet{Hardle2009} while the selection of the optimal bandwidth is still an open issue for the positive convolution approximation of \citet{Bondarenko2003}. 
  
We tested our methodology using weekly index options. We simultaneously estimated  call and put contracts using the data augmentation strategy outlined in Section~\ref{sec_parity}. Our results confirm the appropriateness of the proposed framework. In particular, we observed a good degree of consistency between the inferred SPDs and the index dynamics realized during the residual life of each contract. 

We also tested DESPD performance in a large simulation study using option prices simulated from a 3-component log-normal mixture p.d.f. We compared DESPD with the positive convolution approximation \citep[PCA,][]{Bondarenko2003} and the methods of \cite{Hardle2009} and \cite{Fengler2009}. The PCA framework has been already compared to several alternatives belonging to different modeling classes. The evidence in Section~\ref{sec_simulations} is in favor of our methodology, both in terms of pricing accuracy and density estimation. 

Our future research will focus on generalizations of the present work. We will investigate the asymptotic properties of the proposed estimator. A good starting point should be the results presented in \citet*{Yu2002} and \citet*{Fengler2015}. Our framework can also be adopted within the parametric joint pricing kernel method of \citet[][section 5]{Song2016}. This would enable us to model the impact of multiple factors (e.g. price volatility) on the risk-neutral density. Furthermore, tensor product P-splines \citep{EilersMarx2003} make it possible to model two-dimensional SPDs. One could, for example, account for an intra-day time dimension to model option pricing movements during a trading day. The second dimension could also be the time to maturity \citep[in analogy, for example, with][]{Fengler2015}. This multidimensional DESPD would ensure efficient extrapolation of the SPDs over unobserved maturities and accurate estimation of the prices for longer time to maturity where the observations become more sparse. Finally, our framework does not include an explicit model for  stochastic volatility.  However, the estimated non-parametric density incorporates it implicitly. For explicit estimation of stochastic volatility, time series data would be needed.
%
%%\section*{Acknowledgments}
%%We thank professor Oleg Bondarenko for sharing the code used to estimate risk-neutral densities with the positive convolution approximation. We also thank the two anonymous reviewers and the associated editor for their insightful comments and remarks.
%%
%%--------------------------------------- Disclosure -----------------------------------------%
%%
%% \section*{Disclosure statement}
%% No potential conflict of interest was reported by the authors.
%%
%%---------------------------------------- References ----------------------------------------%
%%
\newpage
\bibliographystyle{natbib}
\bibliography{References}

\end{document}